\documentclass[11pt]{article}
\usepackage{setspace}
\usepackage[margin=2.5cm]{geometry} 
\usepackage{natbib}
\usepackage{graphicx}

\usepackage{amsthm}
\usepackage{amsmath}
\usepackage{amssymb}
\usepackage{mathabx}
\usepackage{multicol}
\usepackage{url}
\usepackage{multirow}
\usepackage{algorithmic} 
\usepackage{algorithm}

\usepackage[inline]{enumitem}

\usepackage{comment}
\usepackage{xcolor}
\newcommand{\ngg}[1]{{\color{black}#1}}
\newcommand{\ds}[1]{{\color{black}#1}}
\newcommand{\ngtwo}[1]{{\color{black}#1}}
\newcommand{\ngrev}[1]{{\color{black}#1}}
\newcommand{\ngrevv}[1]{{\color{black}#1}}

\newcommand{\abs}[1]{\left|#1\right|}
\newcommand{\card}[1]{\abs{#1}}
\newcommand{\R}{\mathbb{R}}

\newtheorem{theorem}{Theorem}
\newtheorem{lemma}{Lemma} 
\newtheorem{corollary}{Corollary} 
\newtheorem{observation}{Observation}
\newtheorem{definition}{Definition}
\newtheorem{remark}{Remark}
\newtheorem{example}{Example}
\newtheorem{proposition}{Proposition}

\theoremstyle{remark}
\usepackage{authblk}

\newcommand{\set}[2]{\left\{#1 \; \left|\;\; #2 \right.\right\}}

\DeclareMathOperator{\U}{P}

\begin{document}

\title{A Parameterized Complexity View on Robust Scheduling with Budgeted Uncertainty}
\author[1]{Noam Goldberg\thanks{goldnoam@bgu.ac.il}} 
\author[2]{Dvir Shabtay\thanks{dvirs@bgu.ac.il}} 
\affil[1,2]{Department of Industrial Engineering and Management, Ben-Gurion University of the Negev, Beer-Sheva, Israel}
\maketitle

\begin{abstract}
In this study, we investigate a robust \ds{single-machine} scheduling problem under processing time uncertainty. The uncertainty is modeled 
using the budgeted approach, where each job has a nominal and deviation processing time, and the number of deviations is bounded by $\Gamma$. The objective is to minimize the 
number of tardy jobs 
\ngrev{where a job is considered tardy if there is some scenario in which it is completed after its due date}. Since the problem is NP-hard in general, we focus on analyzing its tractability under the assumption that 
\ngrev{certain} natural parameter\ngrev{s} of the problem \ngrev{are each} bounded by a constant
. We consider three parameters: the robustness parameter $\Gamma$, the number of distinct due dates in the instance, and the number of jobs with nonzero deviations. 
Using \ds{parameterized}-complexity theory, we prove that the problem is W[1]-hard with respect to $\Gamma$, but can be solved in XP time with respect to the same parameter. With respect to the number of 
\ngrev{distinct} due dates, we establish a stronger hardness result by showing that the problem remains NP-hard even when there are only two different due dates and is solvable in pseudo-polynomial time when the number of due dates is upper bounded by a constant. To complement these results, we show that the case of a common (single) due date reduces to a robust binary knapsack problem with equal item profits, 
\ngrev{a problem} we prove to be solvable in polynomial time. Finally, we prove that the problem is \ds{fixed-parameter tractable with respect to the number of jobs with nonzero deviations.}

\end{abstract}

\section{Introduction}
We study a scheduling problem 
at the intersection of two 
{recent} trends in scheduling research{:} robust scheduling {models} with  
uncertainty  
and parameterized complexity {analysis}. More particularly, we aim to examine whether 
an NP-hard robust scheduling problem becomes tractable when one of its natural parameters (such as the robustness parameter) is of a limited size. Our setting includes the problem of minimizing the number of tardy jobs on a single machine, a problem whose deterministic (also known as the nominal) variant 
can be tractably solved, 
using the classical Moore algorithm~\citep{Moore1968}. 
{This is in {contrast} to the NP-hardness of} the 
robust version of the problem, where {the} processing times {are} subject to budgeted uncertainty, which is the subject of the current paper. 
\ds{Before formally defining our problem and research goals, we briefly introduce the most relevant topics of robust scheduling and parameterized complexity.}

\subsection{Robust Scheduling}
In classical deterministic scheduling models~\citep{Agnetis2025}, it is common to assume that each of the parameters has only a single predefined value. However, 
applications typically involve uncertainty with respect to the input parameters. 
Robust optimization tackles this by defining a set of possible parameter values and optimizing for the worst case within that set. \ds{This set is called the \emph{uncertainty set} of possible parameter values, and each element of this set is called a \emph{scenario}. Based on~\cite{ben2009robust}, robust optimization seeks to address computational challenges in optimization under uncertainty by balancing computational efficiency and probabilistic guarantees under mild assumptions with respect to the distribution of the uncertain parameters.  Such probability bounds may apply even when little distributional information is available.}

\ds{There are several different approaches to define the uncertainty set. One of the most popular approaches is the box uncertainty approach, where each uncertain parameter can assume any value within a given interval (see, e.g.,
~\cite{Dan,Drwal2019,Wang2022}).} Using the box uncertainty approach 
typically leads to an optimal solution at one of the vertices of the box, which is usually considered overly conservative. \ds{Another popular approach is the discrete scenario approach, where the uncertainty set consists of a finite set of scenarios, each explicitly specifying the values of all uncertain parameters (see, e.g.,~\cite{Aloulou2008,Farias,Hermelin1,Gilenson}).} 

\ds{The discrete scenario approach involves the explicit enumeration of all possible or likely scenarios. Excluding scenarios may lead to solutions that are insufficiently conservative. On the other hand, the box uncertainty approach may yield overly conservative solutions, since it is unlikely that all parameters simultaneously attain their worst-case values in an actual realization. To overcome this issue, modern robust optimization approaches have focused on more elaborate definitions of uncertainty sets, such as ellipsoidal uncertainty sets. While motivated by probabilistic guarantees, ellipsoidal uncertainty sets typically lead to intractable combinatorial optimization problems. 

An alternative approach, proposed by~\cite{Bertsimas03} as the budgeted uncertainty approach, involves an uncertainty budget parameter that helps to control the degree of conservatism. Specifically, this model} restricts each uncertain parameter to two possible values: its nominal value or a deviation, where at most $\Gamma$ parameters may deviate from their nominal values. In addition to its intuitive appeal, 
the budgeted uncertainty model is attractive due to {its associated} probabilistic guarantee~\citep{Bertsimas03,bertsimas2004price,bertsimas2021probabilistic}. \ds{Specifically, suppose that the processing times of $n$ jobs are independent, symmetrically distributed random variables drawn from a bounded distribution. Then, $\Gamma$ can be chosen such that, in a $\Gamma$-robust schedule, the tardiness probability of the $j$th early job, for $1 \leq j \leq n$, is bounded above by $e^{-\frac{\Gamma^2}{2j}}$;
see~\cite{Bertsimas03,bertsimas2004price}.  
This} seminal work by Bertsimas and Sim focused on 
tractable combinatorial optimization problems whose robust counterparts remain computationally tractable. It initiated a substantial stream of research within robust combinatorial optimization (see, e.g.,~\ngrevv{\cite{koster2018special,Zeng,busing2023recycling,goerigk2022recoverable}}) including 
studies of the computational complexity of tractable problems whose robust counterpart becomes NP-hard~\citep{goerigk2025Spectrum}, and robust counterparts of problems whose nominal (deterministic) formulation is already NP-hard~\ngtwo{\citep{monaci2013robust,krumke2019robust,bougeret2022constant,goldberg2025robust}}.

\ngrevv{As far as the literature shows,}~\cite{tadayon2011robust} were \ngrevv{the} 
first to apply the 
\ngg{budgeted uncertainty} approach 
to model uncertainty in 
scheduling problem\ngg{s}. Specifically, they 
\ngg{applied it} to 
 the single-machine scheduling problem with the 
 \ngg{objective} of minimizing the total weighted completion time. They considered 
the uncertainty set of processing time 
\ngg{to be} given by 
\begin{equation}
\label{unc}
    \U^{\Gamma} =
    \{p\in\ngrev{\mathbb{N}^n}~:~\delta\in\{0,1\}^n,~~\sum_{j=1}^n\delta_j\leq \Gamma,~~{p}_j=\bar{p}_j+\delta_j\hat p_j~\forall j\in[n]\},
\end{equation}
where $[n]:=\{1,2,\ldots,n\}$ {is the set of jobs}, and 
$\hat{p}_j$ is the amount by which job $j$ can deviate from a nominal value $\bar p_j$. {Accordingly, for each job $j\in[n]$, its actual processing time $p_j$, can be either $p_j=\bar{p}_j$ (the nominal processing time), or $p_j=\bar p_j+\hat p_j$ (the deviated processing time)}. The robustness parameter $\Gamma$ indicates the maximum number of jobs that can deviate from their nominal value in any scenario realization. Any $p \in \U^{\Gamma}$ is a possible scenario of job processing times. 
We note that although the number of possible scenarios is exponential in general, it becomes polynomial when $\Gamma$ is upper {bounded} by a constant.

\cite{Bougeret19} 
{applied the budgeted uncertainty approach to study the robust versions of two classical problems: minimizing the total weighted completion time on a single machine and minimizing the makespan on unrelated parallel machines.}
\ngg{A similar} approach 
\ngg{has been applied more recently to 
determine} robust schedules 
for 
\ngg{other} objective functions (e.g., 
total tardiness and number of tardy jobs) and other machine systems (e.g., 
flow shop and 
job shop), see, e.g.,~\cite{Silva,Bougeret2021,Levorato2022,Bougeret23,Juvin25}.

\subsection{\ds{Parameterized Complexity}}

Parameterized complexity~\citep{DowneyFellowsNew,Niedermeier} is a branch of complexity whose main goal is to analyze the hardness of problems not only in terms of their input size, but also in terms of other ``structural" parameters of the input. This theory has 
recently  
attracted significant interest 
\ngrev{from} the scheduling community (see, e.g., \cite{MnichW15}, \cite{Bevern17}, \cite{Kordon}, and \cite{Her2021}). The central concept in parameterized complexity is that of an FPT algorithm:  

\begin{definition}
An algorithm that solves each instance of a problem in $f(k)n^{O(1)}$ time, where $n$ is the size of the instance, $k$ is the value of the parameter and $f(k)$ is a computable function (monotonically increasing) that depends solely on $k$, is called an FPT algorithm with respect to $k$. The class of problems $\mathcal{FPT}$ is defined as the class of all problems solvable by an FPT algorithm.
\end{definition}

In cases where 
{$k$ is small},
FPT algorithms become very efficient. If we cannot design an FPT algorithm, the next best thing is to design an XP algorithm:

\begin{definition}
An algorithm that solves each instance of a problem in $n^{f(k)}$ time is called an XP algorithm with respect to $k$. The class $\mathcal{XP}$ is the class of all problems solvable by an XP algorithm.
\end{definition}

{Both FPT and XP algorithms \ngg{are guaranteed to} solve 
\ngg{a} given problem in polynomial time whenever $k$=$O(1)$. However, an FPT algorithm is preferable \ngg{as its running time dominates that of XP (asymptotically with respect to $n$).} 
In between these two classes lies the class $\mathcal{W}[1]$. The common assumption used in parameterized complexity is that $\mathcal{FPT} \neq \mathcal{W}[1]$, and this assumption is used to show that certain problems are unlikely to admit FPT algorithms under certain parameterizations. This is done by providing appropriate reductions (called parameterized reductions in the literature) from known $\mathcal{W}[1]$-hard problems (\emph{i.e.} problems for which an FPT algorithm implies $\mathcal{FPT}=\mathcal{W}[1]$, see~\citep{DowneyFellowsNew,Niedermeier} for more information). 
To summarize, we have the following hierarchy between the parameterized complexity classes discussed herein  
    $\mathcal{FPT} \subseteq \mathcal{W}[1] \subseteq \mathcal{XP}$.  
\ds{
\subsection{
\ngrev{Outline of the Paper}}

The remainder of the paper is organized as follows. In Section~\ref{sec:def}, we formally define the robust scheduling problem under study and state our research objective. 
In the three sections that follow, we analyze the parameterized complexity of the problem with respect to three different parameters. The concluding section provides a summary and discussion.}

\section{A Formal Definition of Our Problem and Goals}
\label{sec:def}
\ds{In this section, we formally define the problem, state our research objectives, and establish a key property of optimal schedules that will be instrumental in the parameterized analysis presented later.}
\subsection{Problem Definition}
We 
consider 
a robust single-machine scheduling problem where the objective is to minimize the number of tardy jobs, i.e., the number of jobs completed after their due date. In this problem, we are given a set of $n$ jobs that has to be processed by a single machine without preemption. The uncertainty set of processing times is given by~\eqref{unc}. Moreover, each job {$j\in [n]$} has a (certain) due date, $d_j$, for its completion. We assume that all parameters that define an instance to our problem are restricted to nonnegative integer values. 
A solution of our problem is defined by a schedule, a job permutation $\sigma=(\sigma(1),\ldots,\sigma(n))$, defining the processing order on the single machine. Given such a permutation, $\sigma$, and a possible scenario of job processing time deviation $p \in \U^{\Gamma}$, the \emph{completion time} of job~$j$ is $C_j (\sigma,p) :=\sum_{i \in [n] : \sigma (i) \le \sigma (j)} p_i$; that is, it is the total processing time (under scenario $p$) of all jobs preceding $j$ in the schedule (including $j$ itself). {We use 
$U_j(\sigma,p)$ to 
\ngg{indicate} whether job $j$ is 
{tardy} in schedule $\sigma$ under scenario $p$.} Consequently, {$U_j(\sigma,p)=1$ iff $C_j (\sigma,p)> d_j$ (and 0 otherwise)}. 

We focus on the case where customers must be guaranteed early delivery in advance, before the realization of the actual scenario. Consequently, we say that a job is early in $\sigma$, only if $U_j(\sigma,p)\ngrev{=0}$ for 
\ngrev{every} $p\in \ds{P^{\Gamma}}$. More formally, we define $U_j(\sigma)=\max_{p\in \ds{P^{\Gamma}}} U_j(\sigma,p)$, where job $j$ is considered early if $U_j(\sigma)=0$, and tardy otherwise (i.e., $U_j(\sigma)=1$). The objective is to find a schedule (job processing permutation) $\sigma \in \Sigma$ that minimizes $\sum_{j=1}^n U_j(\sigma)$, where $\Sigma$ denotes the set of all $n!$ possible permutations. Using the classic 3-field notation by~\cite{GrahamLLK79}, we denote our problem by $1|\U^{\Gamma}|\min_{\sigma \in \Sigma} \sum_{j=1}^n U_j(\sigma)$.

\begin{remark}
We note that \cite{Bougeret23} studied a different variant of the same robust single-machine problem with budgeted uncertainty, 
\ngrev{in which the objective is also to} minimize the number of tardy jobs. In their variant, tardy jobs are penalized ex post (i.e., only after the jobs are scheduled and the scenario is revealed). Consequently, there is no need to a priori guarantee the early completion of jobs, and the number of tardy jobs may vary across scenarios. Accordingly, their objective is to find a schedule $\sigma \in \Sigma$ that minimizes $\max_{p \in \U^{\Gamma}} \sum_{j=1}^n U_j(\sigma,p)$. \cite{Bougeret23} proved that this variant is NP-hard. 
\end{remark} 

\subsection{Research Objectives}
As our $1|\U^{\Gamma}|\min_{\sigma \in \Sigma} \sum_{j=1}^n U_j(\sigma)$ problem is NP-hard (see Section~\ref{Robustness}), our main research objective is to analyze its parameterized tractability with respect to three natural parameters of the problem: \begin{enumerate*}[label=(\roman*)]
\item \ngrev{t}he robustness parameter, $\Gamma$, 
\item the number of distinct due dates, $\ngrev{k_d}=\card{\{d_j:j\in [n]\}}$, and 
\item the number of jobs with uncertain processing times, $\ds{k_{u}}=\card{\{j \in [n]: \hat p_j > 0\}}$.\end{enumerate*}

We note that in many practical scenarios the above three parameters have bounded values. 
{The first parameter is relevant, for example, when} the number of jobs that 
deviate 
in any given schedule 
\ngrev{is} small. 
For example, when scheduling medical procedures, $\Gamma$ may correspond to a 
bound on the number of procedures that may take longer due to unforeseen complications~\citep{denton2007optimization,bougeret2022constant}. \ds{The second parameter is also very common in practice, particularly in settings where shipping costs are high. In such cases, the manufacturer is encourage\ngrev{d} to 
\ngrev{offer customers only} a small set of due dates
, 
\ngrev{so that} jobs sharing a due date \ngrev{can be grouped} into a single shipment. The last parameter is practically meaningful in instances where only a few jobs are affected by the source of uncertainty.} 


\subsection{A Key Property of an Optimal Schedule}
\ds{Before 
\ngrev{delving} into the \ngrev{parameterized} complexity analysis, we establish a key structural property 
\ngrev{that at least one} optimal schedule \ngrev{always satisfies}, which will play a central role throughout the paper. \ngrev{Specifically, the next 
lemma shows that our problem exhibits similar properties to those that are known for optimal schedules of the deterministic problem ($1||\min_{\sigma \in \Sigma} \sum_{j=1}^n U_j(\sigma)$), with respect to the ordering of early and tardy jobs.} 
\begin{lemma}
\label{one}
 There exists an optimal schedule for the $1|\U^{\Gamma}|\min_{\sigma \in \Sigma} \sum_{j=1}^n U_j(\sigma)$ problem, where 
 \begin{enumerate}[label=(\roman*)]
     \item \label{item1} The set of early jobs is scheduled before the set of tardy jobs.
     \item \label{item2} The set of early jobs is scheduled in a nondecreasing order of due dates (i.e., according to the EDD rule).
     \item \label{item3} The set of tardy jobs is scheduled in an 
     arbitrary order.
  \end{enumerate}
\end{lemma}
\begin{proof}~
\paragraph{Proof of \ref{item1}.} Consider an optimal schedule $\sigma$ with two adjacent jobs, $j_1$ and $j_2$, where $j_1$ is tardy ($U_{j_1}(\sigma)=1$), $j_2$ is early ($U_{\ngrev{j_2}}(\sigma)=0$), and $j_1$ is scheduled before $j_2$ in $\sigma$\ngrev{, that is $\sigma(j_2)=\sigma(j_1)+1$}. 
\ngrev{C}onstruct 
$\sigma'$ by swapping $j_1$ and $j_2$, while 
\ngrev{leaving all} other jobs 
\ngrev{in place.} 
\ngrev{Then} $j_1$ remains tardy and 
$j_2$ remains early\ngrev{, implying that $\sigma'$ has the same objective value as $\sigma$ so it must be} 
optimal.  
\ngrev{A}pplying this argument repeatedly,  
we end up with an optimal schedule where all early jobs are scheduled before any of the tardy jobs. 
\ngrev{\paragraph{Proof of \ref{item2}.} Consider an optimal schedule $\sigma$ that 
\ngrevv{satisfies} ($i$). Define the lateness of a job $j$, given the schedule $\sigma$ and processing times $p$, as $L_j(\sigma,p)=C_j(\sigma,p)-d_j$. Let $S\subseteq [n]$ denote the set of early jobs in $\sigma$. Accordingly, $L_{j}(\sigma):=\max_{p\in \U^\Gamma} L_j(\sigma,p)\leq 0$ for every $j\in S$. Now, assume \ngrevv{that $\sigma'\neq\sigma$ is a schedule that is similar to $\sigma$, scheduling jobs in $S$ prior to jobs in $[n]\setminus S$,  except that it orders jobs in $S$ according to the EDD rule.} 
Because the worst-case completion time depends on the set of jobs but is independent of their order\ngrevv{, ordering the jobs according to the EDD rule cannot increase the maximum lateness of a job in $S$. In particular, $\max_{j \in S}L_j(\sigma')\leq \max_{j\in S}L_j(\sigma)$. It follows that all jobs in $S$ are early under schedule $\sigma'$.}
}

\paragraph{Proof of \ref{item3}.} \ngrev{This is a straightforward observation that 
follows from the fact that the objective counts tardy jobs and is therefore independent of how tardy they are.}
\end{proof}
}
\section{The Robustness Parameter}
\label{Robustness}
We begin by proving that the $1|\U^{\Gamma}|\min_{\sigma \in \Sigma} \sum_{j=1}^n U_j(\sigma)$ problem is W[1]-hard with respect to the robustness parameter, $\Gamma$. \ds{
\ngrev{Our reduction for proving W[1]-hardness} also serve\ngrev{s as a} standard Karp \ngrev{reduction for showing} NP-hardness 
as the reduction is polynomial, and we reduce from a problem which is known to be NP-hard \ngrev{(as well as W[1]-hard)}. }
We will then complement this hardness result by 
developing an XP algorithm with respect to $\Gamma$. 
\subsection{\ds{W[1]-Hardness and NP-Hardness}}
We prove W[1]-hardness of $1|\U^{\Gamma}|\min_{\sigma \in \Sigma} \sum_{j=1}^n U_j(\sigma)$ with respect to $\Gamma$ 
based on parameterized reduction from the following W[1]-hard problem with respect to the parameter $k$~\ngtwo{\citep{abboud2014losing}} .
\begin{definition} 
$k$-sum problem: Given a set of $h$ positive integers $\mathcal{A}=\{a_1,...,a_h\}$ (where $A=\sum_{a_i \in \mathcal{A}} a_i$), a positive integer $k$ ($1\leq k < h$), and a positive integer $B<A$, determine whether there exists a set $\mathcal{S} \subset \mathcal{A}$ such that $|\mathcal{S}| =k$ and $\sum_{a_i \in \mathcal{S}} a_i=B$.
\end{definition}

Consider the following reduction (which we denote by $\Phi$) from a $k$-sum instance to an instance of $1|\U^{\Gamma}|\min_{\sigma \in \Sigma} \sum_{j=1}^n U_j(\sigma)$ with $n=2h+k+1$ jobs. In the reduction, $M$ can be any integer that {satisfies} $M\geq \max\{2A,h^2\}$. 
\begin{itemize}
\item For each $j\in [h]$, we construct two jobs, $2j-1$ and $2j$. The nominal processing times and the deviations are as follows for each $j\in [h]$:
\begin{align*}
\bar p_{2j-1}&=jM^2  \quad \bar p_{2j}=jM^2+M+a_j; && \hat p_{2j-1}=hM+2a_j \quad \hat p_{2j}=0.
\end{align*}
For {each} $j\in [h-1]$, the due dates are 
\begin{align*}
d_{2j-1}=d_{2j}&= \sum_{j'=1}^{j} j'M^2+jhM+2A \text{, and }\\  
d_{2h-1}=d_{2h}&=\sum_{j'=1}^{h} j'M^2+khM+M(h-k)+A+B.  
\end{align*}

\item For each $j \in \{2h+1,\ldots, 2h+k+1\}$, we construct a single job (job $j$) with processing times
$\bar p_{j}=0$ and $\hat p_{j}=M^3$ 
and with due dates:
$$d_{j}=\sum_{j'=1}^{h} j'M^2+M(h-k)+A-B+(k+1)M^3.$$
\end{itemize}
Finally, we set $\Gamma=k+1$. Define {$O=\{2j-1:j\in[h]\}$} (i.e., the set of all odd jobs in $[2h]$). The following lemma establishes properties of scheduling solutions 
resulting from this reduction, 
\ngg{to be used in} proving 
the problem's hardness. 





\begin{lemma}\label{lem1}
Under the proposed reduction $\Phi$, the following holds 
{for all} schedules of $1|\U^{\Gamma}|\min_{\sigma \in \Sigma} \sum_{j=1}^n U_j(\sigma)$:
\begin{enumerate}[label=(\roman*)]
    \item\label{lem1-1} 
    {A schedule for the} subproblem 
     including only the jobs in $[2j]$ ({for $j\in [h]$}) 
     {has} at most $j$ early jobs. 
     {Further,} any 
     such schedule for {the} subproblem with jobs $[2j]$ { and $j$ early jobs} has exactly one of the two jobs $2j'-1$ (the odd job), or $2j'$ (the even job), early and the other one 
     tardy,
     {for each} $j'\in[j]$.
    \item\label{lem1-0} Every job in $[2h]$ scheduled after a job in $\{2h+1,\ldots,2h+k+1\}$ must be late. Further, (assuming that $h>k+1$) if at most $h$ jobs are late, then it must be that at least $h$ jobs in $[2h]$ are scheduled first, followed by jobs $2h+1,\ldots,2h+k+1$, all \ngrev{of} which are early.    
    \item\label{lem1-2} If $h+k+1$ jobs are early, then at most $k$ (odd) jobs in $O$ are scheduled early.
\end{enumerate}
\end{lemma}
\begin{proof}~ 

\paragraph{\ngrev{Proof of} \ref{lem1-1}.} We 
\ngg{refer to} a subproblem that includes only the jobs in $[2j]$ ($j\in[h]$) as subproblem $j$. The proof of the lemma proceeds by induction on $j$. For the base case ($j=1$), assume, for the sake of contradiction, that there exists a solution for subproblem $1$ with both jobs (1 and 2) being early. Then, the completion time of the second job is at least $\bar p_1+\bar p_2 > 2M^2+M > M^2+hM+2A = d_1=d_2$, which contradicts the assumption that both jobs are scheduled early. Accordingly, in any solution for subproblem $1$, at least one of these jobs must be tardy. Moreover, if there exists a solution with exactly one early job, it must be either job $1$ (the odd job) or job $2$ (the even job). Suppose now that \ngg{ claim}
~\ref{lem1-1} \ngg{of the lemma} holds for $j-1$, so 
for any 
\ngg{schedule} for a subproblem $j-1$, there are at most $j-1$ early jobs, and if there exists a solution with $j-1$ early jobs, it includes one 
of each pair $(2j'-1,2j')$ for any $j'\in[j-1]$ as an early job. Now, by way of contradiction, consider a 
\ngg{schedule} for the subproblem $j$ that has \ds{more than} $j$ early jobs. Based on the 
\ngg{inductive hypothesis}, in such a solution both jobs $2j-1$ and $2j$ are scheduled early. However, the completion time of the last scheduled among these two must be at least
$\sum_{j'=1}^{j-1}j'M^2+2jM^2+M =  \sum_{j'=1}^{j}j'M^2 +jM^2+M> d_{2j-1}=d_{2j}$ in every scenario and for both cases where $j<h$ and $j=h$. This contradicts the fact that both jobs are early, which implies that there are no more than $j$ early jobs in any solution for the subproblem $j$. It further implies that, in any solution with $j$ early jobs, exactly one job from the pair $(2j-1,2j)$ is early (in addition to exactly one job of each of the 
pairs $(2j'-1,2j')$ for $j' \in [j-1]$, 
\ngg{from} the 
\ngg{inductive} hypothesis).\qed


\paragraph{\ngrev{Proof of} \ref{lem1-0}.}  The fact that $M\geq \max\{h^2,2A\}$ implies that for all $\bar{j}\in \{2h+1,\ldots,2h+k+1\}$ and $j\in\ngg{[h]}$, $\hat p_{\bar{j}} = M^3 > 
\sum_{j'=1}^{h} j'M^2+khM+M(h-k)+A+B\ngg{\geq}d_{2j-1}=d_{2j}$. Accordingly, every job \ngtwo{in $[2h]$} that is scheduled after any of the jobs $\bar{j} \in \{2h+1,...,2h+k+1\}$ is {tardy} in any scenario in which job $\bar{j}$ \ngg{deviates}.
\ngg{Now,} consider 
a 
\ngg{schedule of all jobs in $[2h+k+1]$} with at most $h$ {tardy} jobs. Based on~\ref{lem1-1} and the observation about jobs scheduled after jobs in $\{2h+1,\ldots,2h+k+1\}$, these $h$ \ngg{late} jobs 
must be in $[2h]$. Therefore, 
$h$ jobs from $[2h]$ 
\ngg{must be} early and scheduled before the jobs in $\{2h+1,\ldots,2h+k+1\}$, which 
\ngg{follow and} are early as well.\qed

\paragraph{\ngrev{Proof of} \ref{lem1-2}.} If $h+k+1$ jobs are early, by~\ref{lem1-0} 
these must include 
$h$ of the jobs in $[2h]$ as well as the jobs $2h+j$ for $j=1,\ldots,k+1$. Assume, for the sake of deriving a contradiction, that more than $k$ (odd) jobs in $O$ are scheduled early. Then, the completion time of the job in position $h$ in a scenario where $\ngg{\Gamma=}k+1$ jobs 
\ngg{in} $O$ deviate is at least
\[
\sum_{j=1}^hjM^2+(k+1)hM > d_{2h-1}=d_{2h}= \sum_{j=1}^hjM^2+khM+M(h-k)+(A+B). 
\] Therefore, a contradiction is established.
\end{proof}

Based on reduction $\Phi$ from $k$-sum and the resulting schedule properties established by Lemma~\ref{lem1}, we prove the following theorem.

\begin{theorem}\label{thm1}
\ngrev{There is a polynomial-time reduction 
from $k$-sum to 
$1|\U^{\Gamma}|\min_{\sigma \in \Sigma} \sum_{j=1}^n U_j(\sigma)$, where $\Gamma=O(k)$.}
\end{theorem}
\begin{proof}
Given an instance of the $k$-sum problem, $(\mathcal A, B, k)$, we construct 
instance $\Phi(\mathcal A, B, k)$ and 
ask if there is a schedule $\sigma$ 
satisfying $\sum_{j=1}^{n} U_j(\sigma) \leq h$.

\ds{$(\Longrightarrow)$
Suppose} that $(\mathcal A, B, k)$ is a yes-instance of $k$-sum \ngg{with set $\mathcal S$ so that $\card{\mathcal S}=k$ and $\sum_{a\in \mathcal S}a=B$}.  
We construct \ngg{a} schedule $\sigma$ for the $1|\ngrev{\U^\Gamma}|\min_{\sigma \in \Sigma} \sum_{j=1}^n U_j(\sigma)$ instance defined by the reduction $\Phi(\mathcal A,B,k)$ as follows: For $j=1,\ldots,h$, if $a_j \in \mathcal{S}$ we schedule job $2j-1$ in position $j$ (this job will be early) and job $2j$
in position $h+k+1+j$ (this job will be {tardy}). Otherwise \ngg{if $a_j\notin\mathcal S$}, we schedule job 
$2j$ in position $j$ (this job will be early) and job $2j-1$ in position {$h+k+1+j$} (this job will be {tardy}). Moreover, for \ngtwo{each} $j=1,\ldots,k+1$, job $2h+j$ is scheduled in position $h+j$ (this job will be early). 

Let us first prove that the first scheduled $h-1$ jobs in $\sigma$ are completed no later {than} their due dates for every $p \in \U^{\Gamma}$. 
To do so, 
\ngg{first} note that 
the 
completion time of the first $j$ jobs 
in $\sigma$ ($j \in \{1,...,h-1\}$) is at most  
$$
\sum_{j'=1}^{j} j'M^2+jhM+2A\ngg{=d_{2j-1}=d_{2j}. }
$$ 

Consider now the total nominal processing time of the first $h$ jobs scheduled in $\sigma$. Based on the construction, and the fact that $\card{\mathcal S}=k$, this value is 
$$
\sum_{j'=1}^{h} j'M^2+M(h-\card{\mathcal S})+\sum_{a_j\notin \mathcal{S}} a_j=\sum_{j'=1}^h j'M^2 +M(h-k)+A-B.
$$
Moreover, 
the total deviation of the first $h$ scheduled jobs in $\sigma$, following from the fact that $\sum_{a_j\in \mathcal{S}} a_j=B$, 
is at most
$$
hM|\mathcal{S}|+2 \sum_{a_j \in \mathcal{S}} a_j= khM+2B,
$$
noting that 
\ngg{this upper bound is attained} only in the scenario where the $k$ jobs in $S$ are the ones that deviate. 
Therefore, under each of the possible scenarios the completion time of the 
job in position $h$ in $\sigma$ is at most
$$
\sum_{j'=1}^{h} j'M^2+M(h-k)+(A-B)+khM+2B,
$$
which is exactly the due date of the job scheduled in position $h$ (which is either job $2h-1$ or job $2h$). Finally, note that the total nominal processing time of the first $h+k+1$ jobs scheduled in $\sigma$ is
$$
    \sum_{j'=1}^{h} j'M^2+M(h-\card{\mathcal{S}})+\sum_{a_j\notin \mathcal{S}} a_j.
$$
\ngg{Further,} $M^3(k+1)$ is an upper bound on the total deviation of the first $h+k+1$ jobs in any of the possible scenarios (as $\Gamma=k+1$). We note that this bound is attained in a scenario where jobs $2h+1,...,2h+k+1$ {(scheduled in positions $h+1,\ldots,h+k+1$)} are the ones that deviate. Now, using again the fact that $\card{\mathcal{S}}=k$ and that $\sum_{a_j\in \mathcal{S}} a_j=B$, under each of the possible scenarios the completion time of the job in position $h+k+1$  in $\sigma$ is upper bounded by
\begin{align*}
& \sum_{j'=1}^{h} j'M^2+M(h-\ngtwo{\card{\mathcal S}})+\sum_{a_j\notin \mathcal{S}} a_j +M^3(k+1)\\
& =\sum_{j'=1}^{h} j'M^2+M(h-k)+ (A-B) +M^3(k+1)\\
& = \ngg{d_{2h+1}=\cdots=d_{2h+k+1}}.
\end{align*}
It follows that the first $h+k+1$ jobs are early in any of the possible scenarios, and therefore there are at most $h$ tardy jobs. 

\ds{$(\Longleftarrow)$
} Now, suppose we have a job processing permutation $\sigma$ for {our} constructed \ngg{scheduling} instance with $\sum_{j=1}^n U_j(\sigma) \leq h$.
Following Lemma~\ref{lem1}-\ref{lem1-0}, this implies that at most $h$ jobs \ngtwo{in $[2h]$} are scheduled after any of the jobs in $\{2h+1,...,2h+k+1\}$, and $\sigma$ has at least $h$ jobs that are scheduled 
\ngg{prior to} any of the jobs in $\{2h+1,...,2h+k+1\}$ \ngg{and all these jobs as well as jobs in $\{2h+1,...,2h+k+1\}$ are early}.  Following Lemma~\ref{lem1}-\ref{lem1-1}, \ngg{for each $j\in [h]$,} exactly one out of each pair $(2j-1,2j)$ 
is early. 

Let $\mathcal{S}$ be the set of $h$ jobs that are scheduled in $\sigma$ before jobs $\{2h+1,...,2h+k+1\}$. Moreover, define its \ngtwo{\emph{odd subset}} $\mathcal{S}'=\ngg{\mathcal S\cap O}$ 
and \ngtwo{\emph{even subset}} $\bar{\mathcal{S}}=\mathcal{S}\setminus \mathcal{S'}$. 
\ngg{Since at most $h$ jobs are late,} following Lemma~\ref{lem1}-\ref{lem1-2}, 
\begin{equation}
\label{eq0}
    |\mathcal{S'}|\leq k.
\end{equation}
\ngg{Since for all $j\in [h]$,} $\bar{p}_{2j-1}+\hat{p}_{2j-1}>\bar{p}_{2j}+\hat{p}_{2j}$ 
the maximum completion time of the job in position $h$ is 
\ngg{attained when $\card{\mathcal S'}=k$}. 
\ngg{Accordingly,} the completion time of the job in position $h$ is at most
\begin{align*}
& \sum_{j=1}^{h} jM^2+((k+1)h-k)M+\sum_{a_j \in \bar{\mathcal{S}}} a_j+2 \sum_{a_j \in \mathcal{S}'} a_j \\
&
=
\sum_{j=1}^{h} jM^2+((k+1)h-k)M+(A-\sum_{a_j \in \mathcal{S'}} a_j)+2 \sum_{a_j \in {\mathcal{S'}}} a_j    \\
&
\leq \sum_{j=1}^{h} jM^2+((k+1)h-k)M+A+B = d_{2h-1}=d_{2h}.
\end{align*}
Therefore,  
\begin{equation}\label{eq1}
\sum_{a_j \in {\mathcal{S'}}} a_j   \leq B. 
\end{equation}
Next, the fact that jobs $2h+1,...,2h+k+1$ 
\ngg{are} early \ngg{in $\sigma$} implies that 
\begin{align*}
&\sum_{j=1}^{h} jM^2+M|\bar{\mathcal{S}}|+\sum_{a_j \in \bar{\mathcal{S}}} a_j+(k+1)M^3\\ &=  \sum_{j=1}^{h} jM^2+M(h-|\mathcal{S}'|)+A-\sum_{a_j \in \mathcal{S}'} a_j+(k+1)M^3& \\ &\leq
\sum_{j=1}^{h} jM^2+M(h-k)+A-B+(k+1)M^3\\
 & =d_{2h+1}=d_{2h
 +2}=...=d_{2h+k+1},
\end{align*}
or equivalently that
\begin{equation}\label{eq2}
M|\mathcal{S}'|+\sum_{a_j \in \mathcal{S}'} a_j \geq Mk+B. 
\end{equation} From 
\eqref{eq0}-\eqref{eq2} \ngg{it follows that {$|S'|=k$} and {that} $\sum_{a_j\in \mathcal S'}a_j= B$, 
\ngrev{we thereby} conclude that we have \ngrev{a} yes-instance of $k$-sum.}
\end{proof}

\ngrev{Following from Theorem~\ref{thm1}, and the fact that the $k$-sum problem is both NP-hard and $W[1]$-hard with respect to $k$, we have the following two corollaries, respectively.

\begin{corollary}
Problem $1|\U^{\Gamma}|\min_{\sigma \in \Sigma} \sum_{j=1}^n U_j(\sigma)$  is NP-hard.  
\end{corollary}

\begin{corollary}\label{cor:w1hardness}
Problem $1|\U^{\Gamma}|\min_{\sigma \in \Sigma} \sum_{j=1}^n U_j(\sigma)$  is W[1]-hard when parameterized by the number $\Gamma$ of possible deviations.
\end{corollary}}
\noindent\ngrev{Following Corollary~\ref{cor:w1hardness}, we do not expect an FPT algorithm for solving our problem 
(unless FPT$=$W[1]), so we focus next on developing an XP algorithm.} 

\subsection{An XP Algorithm
}

We now develop a dynamic program (DP) and prove that it is an XP algorithm for the $1|\U^{\Gamma}|\min_{\sigma \in \Sigma} \sum_{j=1}^n U_j(\sigma)$  problem. 
In particular, the DP is defined in a state space corresponding to the number of early jobs (at most $n$) times the number of subsets of $\Gamma$ largest 
\ngg{processing time} deviations (at most $n^\Gamma$). 
Following 
Lemma~\ref{one}, we re-index the jobs according to the EDD rule, such that $d_1\leq d_2 \leq\cdots\leq d_n$. 
Define $\Pi_j=\set{\pi\in 2^{[j]}}{\card{\pi}\leq \Gamma}$, that is, the set that includes all possible subsets of $[j]$ of cardinality 
at most $\Gamma$. 
We note that there are $O(n^{\Gamma})$ such subsets in each $\Pi_j$. For each $\pi \in \Pi_j$, we also define $\hat{p}_{\min}(\pi)=\min_{\ell \in \pi} \hat{p}_{\ell}$.

Now, consider 
two 
different schedules 
of the set of the first $j$ jobs, both 
{satisfying the EDD order 
(see Lemma~\ref{one})}. 
Assume that both 
schedules 
have an equal number of 
\ngrev{$\ell$} early jobs  
and 
 the same set of \ds{$\min\{\ngrev{\ell},\Gamma\}$} largest deviating early jobs. 
Then, it is straightforward to show that the 
schedule having 
\ngrev{the least} total nominal processing time of \ngrev{the} early jobs dominates the other. We use th\ngrev{is} 
domination rule to solve the problem \ds{using a DP approach}. To do  so, we define $F_j(\ell,\pi)$ as the minimum total nominal processing time of the set of early jobs among all 
\ngrev{schedules of} the subset of first $j$ jobs with $\ell$ early jobs and with \ds{the same set $\pi$ of \ds{$\min\{\ngrev{\ell},\Gamma\}$} largest deviating early jobs.} \ngrev{
For convenience in evaluating 
$F_j(\ell,\pi)$,} 
define 
\begin{equation}
    C(j,\ell,\pi,\pi'):=F_{j-1}(\ell-1,\pi')+\bar{p}_j +\sum_{i\in \pi} \hat{p}_i.
\end{equation}
\ngrev{That is, $C(j,\ell,\pi,\pi')$ is} the maximal completion time of job $j$ (over all possible scenarios), when 
\ngrev{scheduled} as the next early job 
\ngrev{following a sub-}schedule correspond\ngrev{ing} to state $(\ell-1,\pi')$ at stage $j-1$ of the recursion. Here, $\pi'$ stands for the set of \ds{$\min\{\ngrev{\ell}-1,\Gamma\}$} largest deviating early jobs prior to the inclusion of job $j$ in the set of early jobs. 
\ngrev{S}uch an inclusion is possible only if $C(j,\ell,\pi,\pi') \leq d_j$
\ngrev{, which allows} job $j$ 
\ngrev{to} be an early job. To make sure that $\pi$ matches its definition, we need to select $\pi'$ out of the set $\Pi(\pi,j,\ell)$ (see \ngrev{the} proof of Theorem~\ngrev{\ref{thm:xp}} 
below), where 
\begin{equation}
\label{set}
\Pi(\pi,j,\ell)=\begin{cases}\{\pi \setminus \{j\} \} & \text{if } \ell \leq \Gamma\\
\{\pi \setminus \{j\}\cup \{i\} : i\in[j-1]\ngrev{\setminus \pi} \text{ s.t. } \hat{p}_i \leq \hat{p}_{\min} (\pi)\},  & \text{if } \ell>\Gamma. 
\end{cases} \end{equation}
Using the above definitions, we can use the following recursion to compute $F_j(\ell,\pi)$ for any $j\in[n]$, $\ell\in [j]$ and $\pi \in \Pi_j$: 
    \begin{equation}\label{eq:dprecursion}
F_j(\ell,\pi) 
= \begin{cases}

\min_{\pi' \in \Pi(\pi,j,\ell)}\left\{
\substack{F_{j-1}(\ell-1,\pi')+\bar{p}_j :\\ 
C(j,\ell,\pi,\pi')
\leq d_j}\right\} & j\in \pi\\

\min\{F_{j-1}(\ell,\pi),F_{j-1}(\ell-1,\pi)+\bar p_j\} & j\notin\pi, \hat p_j \leq \hat p_{\min}(\pi), \card{\pi} = \Gamma, 
C(j,\ell,\pi,\pi)
\leq d_j\\ 
 
F_{j-1}(\ell,\pi) & \text{otherwise,}
\end{cases}
\end{equation}
{where} the initial condition is that 
\begin{equation*}
F_0(\ell,\pi)=\begin{cases}
0 & \ell=0 \text{ and } \pi =\emptyset\\
\infty & \text{otherwise.}
\end{cases}
\end{equation*}  
Then, an optimal solution can be determined by
\begin{equation}\label{eq:xpdp}
\max\{\ell\in[n] : \min_{\pi \in \Pi_n}\{F_n(\ell,\pi)\ngtwo{\}} < \infty\}.
\end{equation}


\begin{theorem}\label{thm:xp}
    The DP method given by~
    \eqref{eq:xpdp} 
    \ngrev{determines an optimal schedule for the} \\ $1|\U^{\Gamma}|\min_{\sigma \in \Sigma} \sum_{j=1}^n U_j(\sigma)$ problem in $O(n^{\Gamma+3})$ time. 
\end{theorem}
\begin{proof}
\ngrev{The DP maintains, for each state $(\ell,\pi)$, the least nominal time for completing \ngrev{the} $\ell$th early job given that all jobs in $\pi$ deviate. 
We c}onsider 
\ngrev{three} cases 
\ngrev{for} the state $(\ell,\pi)$. Case 1: 
$j \in \pi$. Note that by definition, if $\ell \leq \Gamma$ then $|\pi|=\ell$, and 
all $\ell$ early jobs 
must be in $\pi$. 
In this case, $(\ell,\pi)$ {can be reached} at stage $j$ only from state $(\ell-1,\pi')$ at stage $j-1$ with $\pi'=\pi \setminus \{j\}$. Otherwise, $\ell > \Gamma$, and $|\pi|=\Gamma$,  
\ngrev{in which case 
we} can reach state $(\ell,\pi)$ at stage $j$ from any state $\pi'= \pi \cup \{i\} \setminus \{j\}$ in {stage} $j-1$ with $\hat{p}_i \leq \hat{p}_{\min} (\pi)$ 
{for} $i<j$. Therefore $\pi'\in \Pi(\pi,j,\ell)$ as defined in~(\ref{set}).

Constructing a feasible solution 
{in} state $(\ell,\pi)$ at
 stage $j$ 
{from} a feasible schedule 
{in} state $(\ell-1,\pi')$ is 
{possible} only if it 
{allows for} early completion of 
job $j$, that is only if $\ngrev{C(j,\ell,\pi,\pi') 
=}F_{j-1}(\ell-1,\pi')+\bar{p}_j +\sum_{i\in \pi} \hat{p}_i\leq d_j$.  
{It} follows 
that 
in Case 1, 
$$
F_j(\ell,\pi) 
= \min_{\pi' \in \Pi(\pi,j,\ell)} \left\{
F_{j-1}(\ell-1,\pi')+\bar{p}_j : C(j,\ell,\pi,\pi') \leq d_j\right\},
$$
where by convention $\min\{\emptyset\}=\infty$.

\ngrev{Case 2:} $j \notin \pi$, $\hat{p}_j \leq  \hat{p}_{\min}(\pi)$ and $\card{\pi}=\Gamma$ \ngrev{and  $C(j,\ell,\pi,\pi) \leq d_j$}. Here, 
 we can 
 reach state $(\ell,\pi)$ at stage $j$ either 
from state $(\ell,\pi)$ \ngrevv{(at stage $j-1$)} by including job $j$ in the set of tardy jobs {(Case 2.1)}, 
or from state $(\ell-1,\pi)$ \ngrevv{(at stage $j-1$)}  by including job $j$ in the set of early jobs {(Case 2.2)}, \ngrev{that is 
$
F_j(\ell,\pi) 
= 
\min\{F_{j-1}(\ell,\pi),F_{j-1}(\ell-1,\pi)+\bar p_j\}. 
$}
\ngrev{Case 3:} $j \notin \pi$ and \ngrev{and at least one of the following holds:} 
$\hat{p}_j > \hat{p}_{\min}(\pi)$ or 
 $|\pi|<\Gamma$ \ngrev{ or $C(j,\ell,\pi,\pi) > d_j$}
. 
Here, we can reach state $(\ell,\pi)$ at stage $j$ only from the same state at stage $j-1$ 
\ngrev{and this is possible only by} including job $j$ in the set of tardy jobs. 
\ngrev{Accordingly,} in Case 3 we must have $F_j(\ell,\pi)=F_{j-1}(\ell,\pi).$ 

\ds{To prove the running-time complexity bound, observe that 
\ngrev{at each stage} $j$, the number of states $(\ell,\pi)$ is $O(n^{\Gamma+1})$, as there are $O(n)$ possible values of $\ell$ and $O(n^\Gamma)$ sets in $\Pi_j$. 
\ngrev{In the worst case corresponding to} 
Case~1, each state can be evaluated in $O(n)$ time \ngrev{(which may be required to iterate through the elements of the set $\Pi(\pi,j,\ell)$)}. \ngrev{The stated complexity bound follows once accounting for the $n$ stages.}}
\end{proof}

\section{The Number of 
\ngrev{Distinct} Due Dates}
In this section, we study the tractability of the $1|\U^{\Gamma}|\min_{\sigma \in \Sigma} \sum_{j=1}^n U_j(\sigma)$ problem when \ds{parameterized} by the number of 
\ngrev{distinct} due dates in the instance, $\ngrev{k_d}$. We begin by showing that the problem is NP-hard already 
{for} $\ngrev{k_d}=2$. We {then} complement our analysis by showing that the problem is solvable in polynomial time when all jobs share a common due date, 
{that is for} 
$\ngrev{k}_d=1$, and develop a pseudo-polynomial algorithm in case of a number of due dates that is constant with respect to the problem's size. 

\subsection{\ngrev{NP-H}ardness for Two Different Due Dates}
{We show that the problem with $\ngrev{k}_d=2$ is NP-hard} 
based on {a} reduction from the NP-hard equal-cardinality partition problem. 
\begin{definition}
Equal-cardinality partition problem: Given a set of $2h$ positive integers $\mathcal{A}=\{a_1,...,a_{2h}\}$ ($\sum_{a_i \in \mathcal{A}} a_i = 2A$), determine whether there exists a partition of $\mathcal{A}$ into two subsets $\mathcal{A}_1$ and $\mathcal{A}_2$ such that $|\mathcal{A}_1|=|\mathcal{A}_2|=h$ and $\sum_{a_i \in \mathcal{A}_1} a_i=\sum_{a_i \in \mathcal{A}_2} a_i=A$.
\end{definition}

The following theorem 
{establishes} that the problem is not in XP (unless P=NP) when \ds{parameterized} by the number of 
\ngrev{distinct} due dates.
\begin{theorem}\label{thm4}
The $1|\U^{\Gamma}|\min_{\sigma \in \Sigma} \sum_{j=1}^n U_j(\sigma)$ problem is NP-hard even {if $\ngrev{k}_d=\card{\set{d_j}{j\in [n]}}=2$.} 
\end{theorem}
\begin{proof}
Given an instance of Equal-cardinality partition problem, we construct the following instance for the $1|\U^{\Gamma}|\min_{\sigma \in \Sigma} \sum_{j=1}^n U_j(\sigma)$ problem with $3h$ jobs and two different due dates. We set $\Gamma=h$. For $j=1,\ldots,2h$, we set $\bar{p}_{j} =M-a_j$ and $\hat{p}_{j} = 2a_j$. For $j=2h+1,...,3h$, we set $\bar{p}_{j}=0$ and $\hat{p}_{j}=M^2$. For $j=1,...,2h$, we set
$$
d_{j}=d^{(1)}=hM+A 
$$ and for $j=2h+1,\ldots,3h$, we set
$$
d_{j}=d^{(2)}=h(M+M^2)-A.
$$
\ngg{For convenience,} we denote $\mathcal{J}_1=[2h]$ 
and $\mathcal{J}_2=\{2h+1,\ldots,3h\}$. To complete the reduction, we select 
\ngg{any} $M>2A$, and ask whether there is a schedule $\sigma$ with $\sum_{j=1}^n U_j(\sigma) \leq h$. 


\ds{$(\Longrightarrow)$} Suppose that the answer for the equal-cardinality partition problem is \ngg{'yes', 
implying that} there exists \ngg{an} $\mathcal A_1$ such that $\card{\mathcal A_1}=h$ 
{and} $\sum_{a\in\mathcal A_1}a=A$. Then, we construct the following \ngg{schedule (permutation)} $\sigma$ for the constructed instance of $1|\U^{\Gamma}|\min_{\sigma \in \Sigma} \sum_{j=1}^n U_j(\sigma)$: letting  $\mathcal{E}_1=\{j|a_j\in\mathcal{A}_1\}$ and $\mathcal{E}_2=\{2h+1,\ldots3h\}$, we first schedule the $h$ jobs in $\mathcal{E}_1$ (in any arbitrary order), followed by the $h$ jobs in $\mathcal{E}_2$ (in any arbitrary order). All other jobs (the jobs in $\mathcal{J}_1\setminus \mathcal{E}_1$) are scheduled last (again, in any arbitrary order).

Consider the first $h$ jobs in the schedule (jobs in $\mathcal{E}_1$). As $\Gamma=h$, 
$$
\max_{p\in \U^{\Gamma}} \sum_{j\in\mathcal E_1}
\ngg{p_j} =\sum_{j\in\mathcal E_1}(\bar p_j+\hat p_j) =  
hM+\sum_{a_j\in \mathcal{A}_1} a_j=hM+A=d^{(1)},
$$
so that the first $h$ jobs are early. Consider now the completion time of the next $h$ jobs scheduled, jobs in $\mathcal E_2$. 
As these jobs have the largest deviation out of the first $2h$ jobs scheduled and $\Gamma=h$, 
\begin{align*}
\max_{p\in \U}\sum_{j\in\mathcal E_1\cup\mathcal E_2}\ngg{p_j}
&= \sum_{j\in \mathcal E_1}\bar p_j+\sum_{j\in \mathcal E_2}(\bar p_j+\hat p_j)\\ & \leq h(M+M^2)-\sum_{a_j\in \mathcal{A}_1} a_j\\ & =h(M+M^2)-A= d^{(2)}, 
\end{align*}
implying that at least $2h$ jobs are early in the constructed instance of the scheduling problem.

\ds{$(\Longleftarrow)$} Now, suppose we have a job processing permutation $\sigma$ 
for 
{our} constructed \ngg{scheduling} instance {with} 
$\sum_{j=1}^n U_j(\sigma) \leq h$. 
Without loss of generality, we may assume that there are exactly $2h$ early jobs in $\sigma$.  
First, note that as $M>A\geq h$, any job in $\mathcal{J}_1$ scheduled after a job in $\mathcal{J}_2$, completes in time later than 
$
\ngrev{M^2}\ngg{>} hM+A \ngg{=} d^{(1)}$, hence it must be {tardy}. 
\ngg{Further, no more than $h$ jobs in $\mathcal J_1$ can be scheduled early since the completion time of the job in position $h+1$ is at least $(h+1)M>hM+A=d^{(1)}$. }
Thus, exactly $h$ jobs 
\ngg{in} $\mathcal{J}_1$ and all the $h$ jobs \ds{from} $\mathcal{J}_2$ are early. Now, let $\mathcal{E}_1$ be the set of early jobs 
in $\mathcal{J}_1$. 
\ngg{Also,} let $\mathcal{E}=\mathcal{E}_1 \cup \mathcal{J}_2$ be the set of early jobs in $\sigma$. As all $h$ jobs in $\mathcal{E}_1$ are early, we have 
$$ \max_{p\in\U^{\Gamma}}\sum_{j\in \mathcal E_1}\ngg{p_j}
=\sum_{j\in\mathcal E_1}(\bar p_j+\hat p_j) \leq 
hM+\sum_{j \in \mathcal{E}_1}a_j \leq Mh+A=d^{(1)}.
$$
Therefore, we have $\sum_{j \in \mathcal{E}_1}a_j \leq A$. As all jobs in $\mathcal{J}_2$ are also early and these jobs have the largest deviations
, we have that
\begin{align*}
\max_{p\in \U^{\Gamma}}\sum_{j\in \mathcal E_1\cup\ngg{\mathcal J_{2}}}\ngg{p_j}
&=\sum_{j\in\mathcal E_1}\bar p_j+\sum_{j\in\ngg{\mathcal J_2}}(\bar p_j+\hat p_j)\\ &\leq  h(M^2+M)-\sum_{j \in \mathcal{E}_1}a_j \\ & \leq h(M^2+M)-A.
\end{align*}
Therefore, we have $\sum_{j \in \mathcal{E}_1}a_j \ngg{=} A$, and $|\mathcal{E}_1|=h$. 
\ngg{Letting} $\mathcal{A}_1=\{a_j|j\in \mathcal{E}_1\}$ and $\mathcal{A}_2=\{a_j|j\in \ngg{\mathcal{J}_{1}} \setminus \mathcal{E}_1\}$, we have a 
\ngg{'yes'-instance of} the equal-cardinality partition problem.
\end{proof}

\subsection{Common Due Date -- Robust Knapsack with Equal Profits}

{We now consider the problem $1|\U^{\Gamma}|\min_{\sigma \in \Sigma} \sum_{j=1}^n U_j(\sigma)$} in the special case of a common due date for all jobs,  
$d$. 
To this end, we consider formulating the problem as a mathematical program. In the deterministic case (where 
$\hat p=\mathbf{0}$), our problem is a special case of the well-studied binary knapsack problem, where all items share the same profit value. For general $\U^\Gamma$ (where $\hat p\neq \mathbf{0}$), our robust scheduling problem is a special case of the robust knapsack problem with equal profits and uncertain cost coefficients under budgeted uncertainty; the 
case with general (not necessarily equal) profits has been introduced in~\cite{monaci2013robust}. 
Thus, our problem is 
\begin{subequations}
\label{prob:robustknapsack}
\begin{align}
& \max_{z\in \{0,1\}^n} && \sum_{j\in [n]} z_j\\
& \text{subject to} && \sum_{j\in [n]}\bar p_j z_j+ 
\max_{p\in\U^\Gamma}\sum_{j\in [n]}(p_j-\bar p_j)z_j\leq d\label{eq:knapsack}.
\end{align}
\end{subequations} 

\ds{Applying a standard dualization technique to the robust constraint~\eqref{eq:knapsack} yields the following formulation, which is known to be equivalent to~\eqref{prob:robustknapsack}:}
\begin{subequations}\label{prob:robustknapsackeqprof}
\begin{align}
& \max_{\ngrev{\mu\in\R_+,}\lambda\in\R^n_+,z\in\{0,1\}^n} && \sum_{j\in [n]} z_j\\
& \text{subject to} && \sum_{j\in [n]}\bar p_j z_j+ \sum_{j\in [n]}\lambda_j+\Gamma\mu \leq d\label{eq:robdual1}\\
& && \hat p_jz_j-\lambda_j\leq \mu  && j\in [n]\label{eq:robdual2}.
\end{align}
\end{subequations} 
The analysis leading to formulation~\eqref{prob:robustknapsackeqprof} is well known and can be found, for example, with a more general objective function in~\citep{Bertsimas03,monaci2013robust}.  Here it can be observed that the constraints~\eqref{eq:robdual1} and~\eqref{eq:robdual2} can be {rewritten} more compactly {as a single (nonlinear) constraint 
\ngg{with} a single \ngg{scalar} dual variable $\mu\geq 0$}.
\begin{align}\label{eq:equivcons}
\sum_{j\in [n]}\bar p_j z_j+ \sum_{j\in [n]}\max\{\hat p_jz_j-\mu,0\}+\Gamma\mu \leq d.
\end{align}
\ngg{For convenience, let $\lambda(\mu)$ denote an $n$-dimensional vector whose component $j$ is defined by $\lambda(\mu)_j=\max\{\hat p_j -\mu,0\}$. Then, for binary $z$, 
constraint~\eqref{eq:equivcons} can be equivalently written as $\sum_{j\in [n]}(\bar p_j+\lambda(\mu)_j)z_j\leq d-\Gamma\mu$. So, for a fixed value of $\mu$,~\eqref{prob:robustknapsackeqprof} amounts to a standard (deterministic) knapsack problem with modified weights (processing times) and equal profit coefficients}. \ngg{Hence, reformulating~\eqref{prob:robustknapsackeqprof} with~\eqref{eq:equivcons}, in place of constraints~\eqref{eq:robdual1} and~\eqref{eq:robdual2}, leads to the following observation.
\begin{observation}\label{obs:itemselect}
For any fixed value $\mu\geq 0$, a solution that is optimal to~\eqref{prob:robustknapsackeqprof} 
can be determined by selecting a maximal set of jobs (starting with $z=\mathbf{0}$ and setting $z_j=1$) following a nondecreasing order of 
processing times $\bar p_j+ \lambda(\mu)_j$
for all $j\in[n]$, as long as~\eqref{eq:equivcons} holds.
\end{observation}
This observation} implies that an optimal solution can be completely determined by correctly ``guessing'' the value of $\mu$. The following lemma establishes that it suffices to consider only $n+1$ possible values for $\mu$. 
\begin{lemma}\label{lem:dualvar}
There exists a solution that is optimal for~\eqref{prob:robustknapsackeqprof} where $\mu = \hat p_k$ for some $k\in [n]$ or $\mu=0$. 
\end{lemma}
\begin{proof}
First observe that we may restrict our attention to $\lambda_j=(z_j\hat p_j-\mu)_+$ for $j\in  [n]$ (as in~\eqref{eq:equivcons}). So, it suffices to consider $0\leq \mu\leq \max_{j\in[n]}\{\hat p_j\}$. 
Now, define $\hat p_0=0$ and assume that for $(z,\mu,\lambda)$ that is optimal to~\eqref{prob:robustknapsackeqprof}, $\hat p_\ell<\mu<\hat p_{\ell+1}$ for some $\ell\in [n-1]\cup\{0\}$. Next, observe that for every binary $z$ we may assume without loss of generality that $\lambda_j=\max\{\hat p_jz_j-\mu,0\}$, for $j\in[n]$ as in~\eqref{eq:equivcons}. Now, replacing $\mu$ by $\hat p_\ell$, so that letting $\lambda'_j=\max\{\hat p_jz_j-\hat p_\ell,0\}\leq \max\{\hat p_jz_j-\mu,0\}=\lambda_j$, 
satisfies
\[
\sum_{j\in [n]}\lambda'_j + \Gamma \hat p_\ell =  \sum_{j:\hat p_j>\hat p_\ell}(\hat p_j-\hat p_\ell) + \Gamma \hat p_\ell \leq \sum_{j\in [n]}\lambda_j + \Gamma \mu \leq d-\sum_{j\in [n]}\bar p_jz_j.\qedhere
\] 
\end{proof}

\noindent Note that Lemma~\ref{lem:dualvar} together with the analysis that follows  resemble\ngrevv{s the} analysis in~\cite[Theorem 3]{Bertsimas03}, where an algorithm is developed for a robust combinatorial optimization problem with uncertain objective coefficients. 

\begin{algorithm}
\caption{Dual search algorithm\label{alg:simplesearch}}
\begin{algorithmic}[1]
\REQUIRE{$\bar p$, $\hat p$, $\Gamma$ {and $d$}.}
    \STATE Sort the jobs in nondecreasing (SPT) order of $\bar p_j+\hat p_j$ for $j\in [n]$ (denoting this ordering $\sigma(0)$).  Sort the jobs in nondecreasing order of $\bar p_j$ for $j\in [n]$ (denoting this ordering $\sigma(n)$).\label{step:initstep}
    \STATE Select early jobs according {to} 
    \ngg{Observation~\ref{obs:itemselect} using} ordering $\sigma(0)$ and denote this solution $z(0)$.
    \FOR{$\mu'\in\set{\hat p_j}{j\in [n]}\ngrev{\cup\{0\}}$}
    \STATE Merge orderings $\sigma(0)$, for $j$ with $\lambda(\mu')_j>0$, and $\sigma(n)$, for $j$ with $\lambda(\mu')_j=0$, to compute an ordering of the jobs in nondecreasing order of $\bar p_j+\lambda(\mu')_j$ for $j\in [n]$ (denoting this order $\sigma(\mu')$).
    \STATE 
    Select \ngg{early} jobs according to \ngg{Observation~\ref{obs:itemselect} using ordering} $\sigma(\mu')$, \ngg{denoting this solution by} $z(\mu')$.\label{step:spt}
    \ENDFOR
    \ENSURE $z$ with maximum $\sum_{j\in [n]}z(\mu)_j$ over all values of $\mu\in \set{\hat p_j}{j\in [n]}\cup\{0\}$. 
\end{algorithmic}
\end{algorithm}

\ngrev{Of course,} \ngg{``naively''} searching for an optimal value of $\mu$ with a complexity bound of $O(n^2\log n)$ is straightforward by re-sorting the jobs for each of the possible values of $\mu$, following Lemma~\ref{lem:dualvar}. \ngrev{Intuitively, Algorithm~\ref{alg:simplesearch} avoids the $O(n\log n)$ cost of resorting the jobs for each possible value of $\mu$. This is done by precomputing two extreme orderings of the values $\bar p_j+(\hat p_j-\mu)_+$ for $j\in [n]$, one for $\mu=0$, and another \ngrevv{for} $\mu=\hat p_{\max}=\max_{j\in [n]}\hat p_j$. Then, any ordering of the values $\bar p_j+(\hat p_j-\mu)_+$ for $0\leq \mu\leq \hat p_{\max}$ can be determined by merging these two orderings in $O(n)$ time.} 
{The following theorem establishes that \ngrev{our} faster implementation of 
\ngrev{the overall} search 
\ngrev{method} is possible, \ngg{as given by Algorithm~\ref{alg:simplesearch}}.} 
\begin{theorem}\label{thm:dualalgcor}
Algorithm~\ref{alg:simplesearch} outputs an optimal solution of~\eqref{prob:robustknapsackeqprof} in $O(n^2)$ time.
\end{theorem}
\begin{proof}
The correctness of this procedure follows from the fact that for a given value of $\mu$ the maximum number of early jobs is found by the ordering in \ds{S}tep~\ref{step:spt} with 
$\bar p_{\sigma(\mu)_1}+\lambda(\mu)_{\sigma(\mu)_1}\leq \cdots \leq \bar p_{\sigma(\mu)_n}+\lambda(\mu)_{\sigma(\mu)_n}$. 
Then, 
$z(\mu')$ is selected so that it is optimal for~\eqref{prob:robustknapsackeqprof} for a fixed value of $\mu$ (with the additional constraint $\mu=\mu'$). Finally, optimality to~\eqref{prob:robustknapsack} follows \ds{by} Lemma~\ref{lem:dualvar} and having selected the best objective value for $\mu'\in \set{\hat p_j}{j\in [n]}\cup \{0\}$.

{\ngtwo{Considering the} running time of the algorithm, 
\ngtwo{the algorithm first performs (in Step~\ref{step:initstep})} two sorting operations of the $[n]$ items which requires $O(n\log n)$ time. Based on Observation~\ref{obs:itemselect}, the selection procedure in  Step 2 requires $O(n)$ time.
Merging of two orderings (sorted lists) in Step 4 into a single sorted list can be done in $O(n)$, {
\ngtwo{as is} the time required for selecting the jobs in Step 5, based on Observation~\ref{obs:itemselect}}. The fact that we repeat these two steps $O(n)$ times (once for each possible value of $\mu'$) leads to a complexity bound of $O(n^2)$. }
\end{proof}

\subsection{A Constant Number of Due Dates}\label{sec:constantduedates}
Suppose that there are $\ngrev{k_d}$ 
\ngrev{distinct} due date values, $d^{(1)},\ldots, d^{(\ngrev{k_d})}$, such that $d^{(1)}\leq \cdots\leq  d^{(\ngrev{k_d})}$. 
Denote the set of jobs having due date $d^{(l)}$, for $l\in [\ngrev{k_d}]$, as $\mathcal J_l$ and further assume that the jobs are indexed according to an EDD order; so that for all $j_1\in \mathcal J_{l_1}$ and $j_2\in \mathcal J_{l_2}$ \ngrev{for} $l_1,l_2\in [\ngrev{k_d}]$, $l_1 < l_2$ implies that $j_1 < j_2$. 
The problem is formulated below as 
\ngrev{a generalization} of~\eqref{prob:robustknapsackeqprof} \ngrev{to have $k_d\geq 1$ 
\ngrev{distinct} due dates}. \ngrev{Note that this problem is written directly in the dualized form with respect to the uncertainty set for each $l\in [\ngrev{k_d}]$ (the same dualization is applied to each constraint independently in order to switch from constraints of the form~\eqref{eq:knapsack} 
to the dualized form~\eqref{eq:robdual1} with dual variable $\mu_l$ for each $l\in [k_d]$).} 
\ngrev{The resulting formulation is}
\begin{subequations}\label{prob:kduedates}
\begin{align}
& \max_{\mu\in\R_+^{k_d},\lambda\in\R_+^{k_d\times n},z\in \{0,1\}^n} && \sum_{j\in [n]}z_j\\
& \text{subject to} && \sum_{j\in \bar{\mathcal J_l}} (\bar p_jz_j+\lambda_{lj}) + \Gamma\mu_l \leq d^{(l)} && l\in [\ds{k_d}]\label{cons:kduedate}\\
& && \hat p_jz_j-\lambda_{lj}\leq \mu_l && j\in \bar{\mathcal J_l}, l\in [\ngrev{k_d}], 
\end{align}
\end{subequations}
where $\bar{\mathcal J_l}=\bigcup_{i=1}^l\mathcal J_i$. In 
\ngrev{optimal solutions $(\mu,\lambda,z)$ of this (dualized)} formulation, \ngrev{the dual multipliers} $\mu_l$ 
\ngrev{are} the minimal deviations among all early jobs that deviate in the worst-case scenario \ngrev{for} 
jobs in $\bar{\mathcal J_l}$. \ngrev{Further, it} \ds{
\ngrev{can be observed} that there exists an optimal solution for~\eqref{prob:kduedates} where $\lambda_{lj}=\ngrev{(\hat p_jz_j-\mu_l)_+}$. This value \ngrev{
is then} the amount that \ngrev{an early} job $j\in \bar{\mathcal J_l}$ deviates \ngrev{in excess of} $\mu_l$\ngrev{. 
Similar to Lemma~\ref{lem:dualvar} we have in this case the following proposition.
\begin{proposition}\label{prop:kduedates}
There exists a solution $(\mu,\lambda,z)$ that is optimal for~\eqref{prob:kduedates} and that satisfies for each $l\in [k_d]$, $\mu_l = \hat p_j$ for some $j\in \bar{\mathcal J}_l$, or $\mu_l=0$. 
\end{proposition} \noindent The proof of this proposition is essentially identical to that of Lemma~\ref{lem:dualvar} once it is separately applied to each constraint $l$, where set $[n]$ is replaced by $\bar{\mathcal J}_l$.}}

In the following we 
\ngrev{provide} several definitions that will \ngrev{be useful in describing} 
\ngrev{the DP for solving~}
~\eqref{prob:kduedates}. For each $l\in [\ngrev{k_d}]$ and $j\in \mathcal {J}_l$, let $p_j^{(l)}=\bar p_j+(\hat p_j-\mu_l)_+$ and $\tilde d^{(l)}=d^{(l)}-\Gamma\mu_l$. Given $\mu_l$, the $p_j^{(l)}$ value represents the 
\ngrev{sum of nominal and deviation (in excess of the minimal deviation)} processing time of job $j\in \mathcal {J}_l$, in the 
worst-case  
scenario \ngrev{corresponding to $\mu_l$}, 
if \ngrev{job} $j$ is 
early. 

For a given (fixed) vector $\mu\in\R_+^{\ngrev{k_d}}$, denote \ngrev{the $l$-\emph{adjusted} processing time of job $j$ by  $\bar p_j+(\hat p_j-\mu_l)_+$ for each $l\in[k_d]$ and $j\in \mathcal J_l$. Moreover, denote the (modified) due dates by $\tilde d^{(l)}$. Then, the DP value function} $f_{\mu}(j,q,v)$ 
\ngrev{is defined} as the minimal 
total \ngrev{$k_d$-adjusted processing time %
of} 
the set of early jobs 
\ngrev{of all schedules of job set $[j]$ having $q$ early jobs, such} that \ngrev{the completion time of the last job in $\bar{\mathcal J}_l$ for each $l\in [k_d-1]$ is 
$v_l$. As an example for the DP states and value function consider the following instance.  
\begin{example}
 Let $\Gamma=1$, $\mathcal J_1=\{1,2\}$ with $d^{(1)}=3$ and $\mathcal J_2=\{3\}$, with $d^{(2)}=4$. Also, let $\bar p_1=\bar p_2=\bar p_3=1$, $\hat p_1=\hat p_3=2$ and $\hat p_2=1$. Then,
$f_{(1,2)}(2,1,1)=\bar p_2+(\hat p_2-\mu_2)_+=1$ as job 2 is the only job in the set $[2]$ that can be scheduled early while satisfying the required completion time of $v=1$ (in particular $f_{(1,2)}(1,1,1)=f_{(1,2)}(0,0,-1)=\infty$), for jobs in $\mathcal J_1$. Then, 
$f_{(1,2)}(3,2,1)=\bar p_2+(\hat p_2-2)_++\bar p_3+(\hat p_3-2)_+=2$ (note that for $\mu=(1,2)$,  $\tilde d^{(1)}=2$ and $\tilde d^{(2)}=2$).
\end{example}}  
\ngrev{For convenience also define, f}or vector $x\in\R^n$ and integers $a,b$ such that $0<a<b\leq n$, 
$(x)_a^b$ as the subvector $(x_a,x_{a+1},\ldots,x_{b})$. 
\ngrev{Next,} define 
\begin{align*}
 M=\ngrev{\prod_{l=1}^{k_d}}\left(\{\hat p_j:j\in \ds{\bar{\mathcal J}_l}\}\cup\{0\}\right) && \text{ and } && {V_\mu} = 
 \ngrev{\prod_{l=1}^{k_d-1} [\tilde d^{(l)}]}.
\end{align*} The set $M$ contains candidate vectors $\mu$, with each possible value of $\mu_l$ in coordinate $l\in [\ngrev{k_d}]$; following 
\ngrev{Proposition~\ref{prop:kduedates}} it suffices to consider $\ngrev{\card{\bar{\mathcal J}_l}}+1$ such values. 
The set $V_\mu$ contains all possible values of 
completion times for each of the early job sets $\ngrev{\bar{\mathcal J}}_l$ for $l\in [\ngrev{k_d}-1]$. 


Note that 
{for any} $v\in V_\mu$ 
{the DP} state corresponds to a schedule that meets the due dates $d^{(l)}$ for ${l}\in[\ngrev{k_d}-1]$. 
Accordingly, {the DP recursion is}
\begin{equation}
\label{two}
f_{\mu}(j,q,v)=\min
\left\{\substack{
    f_{\mu}(j-1,q,v) \\
    f_{\mu}\left(j-1,q-1,((v_i)_{i=1}^{\ngrev{\ell(j)}-1}, (v_i-p^{(\ngrev{i})}_j)_{i=\ngrev{\ell(j)}}^{\ngrev{k_d}-1} ))\right)+p^{(\ngrev{k_d})}_j}
    \right\},
\end{equation}
where $\ell(j)$ 
is \ngrevv{the} index $l\in [k_d]$ so that job $j\in \mathcal J_l$. 
\ds{The relation in~\eqref{two} holds due to the fact that we can 
reach $(j,q,v)$ from $(j-1,q,v)$ by including job $j$ in the set of tardy jobs. Otherwise, if we assign job $j$ to the set of early jobs, we can reach $(j,q,v)$ only from state $(j-1,q-1,((v_i)_{i=1}^{\ngrev{\ell(j)}-1}, (v_i-p^{(\ngrev{i})}_j)_{i=\ngrev{\ell(j)}}^{\ngrev{k_d}-1} ))$ as the total processing time of job $j$ increases the completion time of the last job in $\bar{\mathcal{J}_{l'}}$ by $p^{(\ngrev{l'})}_j$ for $l'>\ngrev{\ell(j)}$. }

Using the recursion in~\eqref{two}, we can compute all $f_{\mu}(j,q,v)$ values for $\mu \in M$, $j \in [n]$, $q\in [j]$ and $v \in V_\mu$, with the base conditions {that 
\[
f_{\mu}({0},q,v)=\begin{cases}
    0 & q=0,v=\mathbf{0}\\
    \infty & \text{otherwise.}
\end{cases}
\]}
Then, the optimal solution can be determined by
\begin{equation}\label{dp:kduedates}
\max \left\{q\in[n]:\min_{\mu\in M,v\in V_\mu}f_{\mu}(n,q,v)\leq \tilde d^{(\ngrev{k_d})} 
\right\}.
\end{equation} 

Finally, the following theorem establishes the runtime complexity of our DP method.
\begin{theorem}\label{thm:kduedates}
The DP algorithm~\eqref{dp:kduedates} evaluates a solution that is optimal to~\eqref{prob:kduedates} in time\\ $O(n^{\ngrev{k_d}+1}\ngrev{\log n}D^{\ngrev{k_d}-1})$ 
\ngrev{where} $D = d^{(\ngrev{k_d}-1)}$.
\end{theorem}
\begin{proof}
The algorithm 
\ngrev{iterates in search} for $q\in [n]$ \ngrev{in~\eqref{dp:kduedates}. For each candidate value of $q$,} 
iterating through $\mu\in M$ is $O(n^{\ngrev{k_d}})$ times the complexity of the DP recursion.  
The running time complexity of the DP is bounded by the number of states $\card{V_{\mu}}\ngrev{\leq} D^{\ngrev{k_d}-1}$, over $n$ stages, 
\ngrev{where} each evaluation of $f$\ngrev{, given by~\eqref{two}, 
is} $O(1)$ since $\ngrev{k_d}$ is a constant\ngrev{, leading to an $O(nD^{k_d-1})$ complexity bound for the DP recursion for each value of $q$ and a given dual multiplier vector $\mu$. The claim follows once accounting for the iterations through the possible values $\mu$ and noting that the maximum value of $q$ can be determined by binary search. In particular,  binary search is possible based on the fact that the worst-case completion time of early jobs, given by $\min_{\mu,v}\{f_\mu(n,q,v)+\Gamma\mu_{k_d}\}$, is nondecreasing in $q$.}  
\end{proof}
\ngrev{The following remark compares the complexity of our DP algorithm 
with that obtained by treating problem~\eqref{prob:kduedates} more generally as a robust multidimensional knapsack problem.
\begin{remark}
Note that for fixed choices of $\mu$, our problem~\eqref{prob:kduedates} is a special case of the multidimensional knapsack problem with unit profit values. Applying the multidimensional knapsack DP algorithm of~\cite{weingartner1966capital} (see also \cite{Kellerer2004})  can be done in $O(nk_d D^{k_d})$ time. This would have led to an overall complexity of $O(n^{k_d+1}k_d \bar D^{k_d})$, where $\bar D=d^{(k_d)}\geq D$, compared with $O(n^{k_d+1}\log n D^{k_d-1})$ of Theorem~\ref{thm:kduedates}. 
\end{remark}}

\section{Number of Jobs with Nonzero Deviations}

Given an instance for $1|\U^{\Gamma}|\min_{\sigma \in \Sigma} \sum_{j=1}^n U_j(\sigma)$, let $\mathcal{J}_1=\set{j \in[n]}{\hat{p}_j \neq 0}$ and $\mathcal{J}_2=\set{j\in[n]}{\hat{p}_j = 0}$
.
Furthermore, let $\ds{k_{u}}=|\mathcal{J}_1|$, {that is,} 
the number of jobs with nonzero deviations (or uncertain processing times). In this section, we take $\ds{k_u}$ as a parameter, and show that the problem is solvable in FPT time with respect to this parameter. Note that, without loss of generality, we can assume that 
{$\ngrev{k_u}>\Gamma$. {Otherwise, in the worst-case scenario, all jobs in $\mathcal{J}_1$ deviate and the problem reduces to a nominal (deterministic) problem with processing times $p=\bar p+\hat p$}.}



To solve the robust problem, we break it \ds{down} into 
$O(2^{\ds{k_{u}}})$ subproblems, corresponding to all possible subsets 
$\mathcal{E}_1\subseteq \mathcal{J}_1$. In each 
{of these} subproblems, we solve a restricted 
variant of the original problem where the jobs in $\mathcal{E}_1$ are restricted to be early, while the {jobs in} \ngrev{the} complementary set 
$\mathcal{J}_1 \setminus \mathcal{E}_1$ are {restricted to be} 
tardy. By solving all $O(2^{\ds{k_{u}}})$ subproblems and 
{selecting} the best solution 
{of all feasible subproblem solutions}, we 
{determine an} optimal solution for the original problem. 
{In the following we assume that} the jobs {are ordered} according to the EDD rule, such that $d_1\leq \cdots \leq d_n$ {(otherwise they can be renumbered to satisfy this ordering)}.

{Consider now a given subproblem, and let 
$\mathcal{E}_1=\{\mathcal{E}_1(1),\ldots,\mathcal{E}_1(k')\}\subseteq \ds{\mathcal{J}_1}$ be the 
\ngrev{EDD-}ordered subset of \ds{uncertain} jobs that are restricted to be early ($k'\leq  \ds{k_{u}}$).} For convenience, for $j\in \mathcal E_1$ define $\mathcal E_1^{-1}(j)$ as its position in the ordering, which must be in $[k']$. 
The {remaining} jobs in $\mathcal{J}_1 \setminus \mathcal{E}_1$ are restricted to be tardy. As all tardy jobs are scheduled last in an arbitrary order, we can omit this set of jobs from the instance, which now includes only $n'=n-\ds{k_{u}}+k'$ jobs. The resulting subproblem is a restricted variant of $1|\ngrev{\U^\Gamma}|\min_{\sigma \in \Sigma} \sum_{j=1}^n U_j(\sigma)$ on $n'\leq n$ jobs where there are $k'\leq \ds{k_{u}}$ {(nonzero deviation)} jobs that are constrained to be early. 

\ds{

Below, we show that, assuming the jobs are indexed according to the EDD rule, each subproblem can be solved in linear time by extending Moore’s algorithm. The main idea behind the algorithm is that, once $\mathcal E_1$ is fixed, the set of uncertain jobs that will be scheduled early is determined exactly (the uncertain set of tardy jobs in $\mathcal J_1 \setminus \mathcal E_1$ can be viewed as rejected jobs and therefore ignored entirely). As there exists an optimal EDD schedule for the set of early jobs (see Lemma~\ref{one}), given $\mathcal E_1$ and $\Gamma$, we know the exact maximum total deviation for every prefix of $\mathcal E_1$. This knowledge enables us to determine the exact remaining 
\ngrev{slot} available for scheduling jobs in $\mathcal{J}_2$ within each prefix set of jobs in $\mathcal E_1$. If any of these 
\ngrev{slots} 
has \ngrev{a}
negative length
, then the corresponding subproblem 
\ngrev{does not admit a} feasible solution (see Section~\ref{feascheck}). Otherwise, as shown in Section~\ref{mooreextention}, we can extend Moore’s algorithm to find a subset of $\mathcal{J}_2$ of maximum cardinality that can be optimally inserted into these 
\ngrev{slots}, thereby 
\ngrev{determining} an optimal schedule for the subproblem. Repeating this procedure for the entire collection of subproblems 
\ngrev{yields} our FPT algorithm (Algorithm~\ref{alg:hmext} below).} 

\subsection{Feasibility Check of a Given Subproblem}
\label{feascheck}
As a first step toward solving a given subproblem, we must test for feasibility; that is, to determine if $\mathcal{E}_1$ is a feasible set of early jobs. 
Now, for each \ds{prefix}  
including the first $\ell\leq k'$ jobs in this set, we compute $\Delta_\ell=\max_{p\in\U^{\Gamma}}\{\sum_{j=1}^\ell (p_{\mathcal E_1(j)}-\bar p_{\mathcal E_1(j)})\}$, which is the maximum total deviation of the $\min\{\ell,\Gamma\}$ jobs with the largest deviation in $\{\mathcal{E}_1(1),\ldots,\mathcal{E}_1(\ell)\}$. (The $\Delta_\ell$ values can be computed in $O(k')=O(\ds{k_{u}})$ time.) \ngrev{Further, for convenience let $\Delta_{0}=0$.}
As the EDD rule {(recall that $d_{\mathcal{E}_1(1)} \leq...\leq d_{\mathcal{E}_1(k')}$)} minimizes the maximal lateness, $\mathcal{E}_1$ is a feasible set of early jobs if and only if 
for {all} $\ell\in[k']$, 
\begin{equation}
\label{check}
    \sum_{i=1}^\ell \bar{p}_{\mathcal{E}_1(i)}+\Delta_\ell \leq d_{\mathcal{E}_1(\ell)}.
\end{equation}
\ds{If condition~(\ref{check}) fails to hold for some $\ell \ngrev{\in} [k']$, then the corresponding subproblem is infeasible. Otherwise, we proceed in Section~\ref{mooreextention} to solve the subproblem to optimality} by 
finding a maximum cardinality subset $\mathcal{E}_2 \subseteq \mathcal{J}_2$ {(the set of certain or zero-deviation jobs)}  
that \ngrev{can be scheduled} \ds{
early} 
together with jobs in $\mathcal E_1$. 
\subsection{Extending Moore's Algorithm for Our Subproblem}
\label{mooreextention}
{We now develop} an algorithm that optimally solves a subproblem 
for a given feasible {sub}set $\mathcal{E}_1$ of {uncertain} early jobs. The algorithm is in fact an extension of Moore's classical algorithm~\citep{Moore1968} to the case where 
\begin{enumerate*}[label=(\roman*)]\item the jobs in $\mathcal{E}_1$ are restricted to be early while 
\ngrev{the} jobs in $\mathcal J_1\setminus \mathcal E_1$ are restricted to be {tardy} 
and \item {the processing times of} 
the jobs in $\mathcal{E}_1$ {may deviate as defined by the budgeted uncertainty set $\U^\Gamma$}.\end{enumerate*} Of course, since jobs in $\mathcal J_2$ have zero deviations, 
only jobs in $\mathcal E_1$ may deviate. 
For each $\ell\in [k']$ we compute the slack (gap) of 
job 
$\mathcal{E}_1(\ell) \in \mathcal{E}_1$ as 
\begin{equation}
\label{slack}
    \delta_{\mathcal{E}_1(\ell)}\ds{:}= d_{\mathcal{E}_1(\ell)}-\sum_{i=1}^\ell \bar{p}_{\mathcal{E}_1(i)}-\Delta_\ell.
\end{equation}
The slack of $\mathcal{E}_1(\ell)$ corresponds to the total processing time of jobs in $\mathcal{J}_2$ that can be scheduled before $\mathcal{E}_1(\ell)$ without causing it to be 
tardy 
(recall that all jobs in $\mathcal{E}_1$ are restricted to be early). Accordingly, the total processing time of all early jobs from $\mathcal{J}_2$ that 
{can} be processed prior to job $\mathcal{E}_1(\ell)$ 
{is at most}
\begin{equation}
\label{slack2}
\bar{\delta}_{\mathcal{E}_1(\ell)}=\min_{i\in \{\ell,\ldots,k'\}}\delta_{\mathcal{E}_1(i)}.
\end{equation} 
\ngrev{For convenience define also $\bar\delta_{\mathcal E_1(k'+1)}=\infty$.}


The algorithm for solving a given subproblem, formally presented in the following as Algorithm~\ref{alg:hmext}, iteratively proceeds on the jobs in $\mathcal J$. In each iteration $j$, it optimally solves a subproblem 
{with} jobs in $[j]\cup\mathcal{E}_1$. Specifically, it finds a subset of jobs in $\mathcal J_2\cap [j]$ 
{of maximum} cardinality that can be scheduled early without violating the restriction that all jobs in $\mathcal{E}_1$ are early. If multiple such 
subsets exist, it selects {as a subproblem optimal solution} the one 
with the minimum total processing time -- we will refer to such a subset as a \emph{compact} subset (or compact solution). Note that, by definition, any compact solution is also an optimal solution for the corresponding subproblem. We note that the algorithm reduces to Moore's algorithm when $\mathcal{E}_1=\emptyset$. In fact, Moore's algorithm 
\ngg{applies} a similar idea of maintaining at each iteration $j$
{a maximum cardinality compact} set of early jobs,  
{ that is a solution that is optimal for scheduling the 
jobs in %
subset $[j]$,  
with the 
least total processing time among all optimal solutions.}

\begin{algorithm}
\caption{{Extended Moore's method} for a {restricted robust} subproblem 
}
\label{alg:hmext}
\begin{algorithmic}[1]
\REQUIRE{$\mathcal{J}_1$, $\mathcal{J}_2$, $\mathcal{E}_1 \subseteq \mathcal{J}_1$, $\bar{p}_j$ and $d_j$ for $j\in [n]$, $\hat{p}_j$ for each job $j \in \mathcal{E}_1$ and $\Gamma$.}
\STATE Initialization: $P^{(0)}\leftarrow 0$, $\mathcal{E}^{(0)}_2\leftarrow \emptyset$, $\ell'\leftarrow 0$, $k'\leftarrow \card{\mathcal E_1}$. Compute $\delta_{\mathcal{E}_1(\ell)}$ and   $\bar{\delta}_{\mathcal{E}_1(\ell)}$ by eqs.~(\ref{slack})-(\ref{slack2}) for $\ell=1,\ldots,k'$.
\IF {(\ref{check}) does not hold for some $\ell\in[k']$} 
\STATE output ``infeasible subproblem'' and exit.
\ENDIF
\FOR{$j=1,\ldots,n$}\label{step:alg2begloop}
\IF{ $j\in \mathcal{E}_1$}
\STATE $\ell'\leftarrow\ell'+1$, $\mathcal{E}_2^{(j)}\leftarrow \mathcal{E}_2^{(j-1)}$ and $P^{(j)}\leftarrow P^{(j-1)}$ 
\ELSE
\STATE $P^{(j)}\leftarrow P^{(j-1)}+\bar p_j$ and $\mathcal{E}_2^{(j)}\leftarrow \mathcal{E}_2^{(j-1)} \cup \{j\}$.
\IF{$P^{(j)}+\sum_{i=1}^{\ell'}\bar{p}_{\mathcal{E}_1(i)} + \Delta_{\ell'} > d_j$ or $P^{(j)} > \bar{\delta}_{\mathcal{E}_1(\ell'\ngrev{+1})}$}\label{algstep:postponemaxnom} 
\STATE compute $j^*\in \arg\max_{\ngtwo{q} \in \mathcal{E}_2^{(j)}} \{\ds{\bar{p}}_{\ngtwo{q}}\}$, set $\mathcal{E}_2^{(j)} \leftarrow \mathcal{E}_2^{(j)} \setminus \{j^*\}$ and update $P^{(j)} \leftarrow P^{(j)}-\ds{\bar{p}}_{j^*}$. 
\ENDIF
\ENDIF
\ENDFOR\label{step:alg2endloop}
\ENSURE{$\mathcal{E}_2^{(n)}$}
\end{algorithmic}
\end{algorithm}

\begin{theorem}\label{thm5}
{Given that $\mathcal{E}_1 \subseteq \mathcal{J}_1$ is a feasible set of 
early jobs, Algorithm~\ref{alg:hmext} determines, at 
iteration $j\ngg{\in[n]}$ \ngg{(}of Steps \ref{step:alg2begloop}-\ref{step:alg2endloop}), a compact set of jobs $\mathcal{E}_2^{(j)} \subseteq \mathcal{J}_2 \cap [j]$ for the subproblem with job set $[j]\cup \mathcal{E}_1$.}
\end{theorem}
\begin{proof}
We prove the theorem by induction on the iteration index, $j$. Let us first prove that the theorem holds for the \ngg{base case} ($j=1$). \ngg{First,} consider 
the case where $1\in \mathcal{J}_1$. 
\ngg{In} this case, $\mathcal{E}_2^{(1)}=\emptyset$ is the only feasible solution (as $\mathcal{E}_2^{(1)}$ is a subset of $\mathcal{J}_2 \cap \{1\}=\emptyset$ when $j=1$). Therefore,  $\mathcal{E}_2^{(1)}
=\emptyset$ is \ds{a} compact solution for this case. Next, consider 
the 
complementary case where $1\in \mathcal{J}_2$. In this case, job 1 can be early only if 
\ngrevv{neither} of the two conditions in Step~\ref{algstep:postponemaxnom} holds. The first condition \ngg{(not being satisfied)} verifies that job 1 can be early if we schedule it first, while not satisfying the second condition ensures that all jobs in $\mathcal{E}_1$ remain early. Conversely, if one of the two conditions 
holds, $\mathcal{E}_2^{(1)}=\emptyset$ is the only feasible solution at the end of the first iteration. Therefore, 
it must be a compact 
solution. 
Otherwise, $1\in \mathcal{J}_2$ and \ngrevv{neither of} the two conditions 
\ngrevv{is} satisfied \ngrevv{(a job's nominal processing time cannot exceed its own due date)}. The facts that $\mathcal{E}_2^{(1)}$ is a subset of $\mathcal{J}_2 \cap \{1\}=\{1\}$, and that
$\mathcal{E}_2^{(1)}=\{1\}$ 
is the unique optimal solution, 
(trivially) \ngrevv{imply} that $\mathcal{E}_2^{(1)}$ is a 
compact solution.

Now, assume that the claim holds at the end of iteration $j-1$ ($j-1 \in 
[n-1]$) and accordingly  $\mathcal{E}_2^{(j-1)}$ is an optimal compact solution 
\ngg{for} job set $[j-1]$ (the induction hypothesis). We next show that the claim holds also at the end of iteration $j$ for job set $[j]\cup \mathcal{E}_1$. If $j \in \mathcal{J}_1$ then $\mathcal{E}_2^{(j)}=\mathcal{E}_2^{(j-1)}$ is a 
compact solution at the end of iteration $j$ by the induction hypothesis. 
So, we focus on the case where $j \in \mathcal{J}_2$. If either of the two conditions in Step \ref{algstep:postponemaxnom} 
holds, we cannot schedule $j$ as an early job without removing one of the jobs in $\mathcal{E}_2^{(j-1)}$. 
It implies that an optimal solution at the end of iteration $j$ has the same cardinality 
as $\mathcal{E}_2^{(j-1)}$  (note that $\mathcal{E}_2^{(j-1)}$ remains optimal for job set $[j]$ although it may not be a compact 
solution for the subproblem with job set $[j]\cup\mathcal{E}_1$). 
$\mathcal{E}_2^{(j-1)}$ \ngg{being a} 
compact solution at the end of Iteration $j-1$ \ngg{(by the induction hypothesis)}, 
implies that if $\mathcal{E}_2^{(j-1)} \cup \{j\} \setminus \{j^*\}$ is a feasible set of early jobs, it is also a 
compact solution at the end of iteration $j$. 
Note that all jobs in $\mathcal{E}_2^{(j-1)} \setminus \{j^*\}$ are completed not later than the completion time of these jobs in $\mathcal{E}_2^{(j-1)}$. Therefore, they are all scheduled early at the end of iteration $j$. It remains to show that 
\ngg{for} $\mathcal{E}_2^{(j)}=\mathcal{E}_2^{(j-1)} \cup \{j\} \setminus \{j^*\}$, job $j$ (
scheduled last among all jobs in {$\mathcal{E}_2^{(j)}$}) will be an early job. Let $j'$ be the last scheduled job in $\mathcal{E}_2^{(j-1)}$ and let 
$C^{(j-1)}_{j'}$ be 
its completion time under solution $\mathcal{E}_2^{(j-1)}$. 
Now, let $C^{(j)}_{j}$ be the completion time of the last scheduled job in $\mathcal{E}_2^{(j)}=\mathcal{E}_2^{(j-1)} \cup \{j\} \setminus \{j^*\}$. It follows that 
$$C^{(j)}_{j}=C^{(j-1)}_{j'}-\ngg{\bar p_{j^*}}+\ngg{\bar p_j}\leq C^{(j-1)}_{j'}\leq d_{j'}\leq d_j,$$
where the second inequality {follows} \ngg{from} $\mathcal{E}_1 \cup \mathcal{E}_2^{(j-1)}$ {being a feasible set of early jobs} for $[j-1]$, and the last inequality {follows} from the EDD {ordering of the jobs}.
Hence, $\ngrev{\mathcal{E}_{2}^{(j)}}$ is a 
compact solution. Lastly, if 
\ngrevv{neither} of the two conditions of Step 
\ref{algstep:postponemaxnom} holds, the optimal solution value is $|\mathcal{E}_2^{(j-1)}|+1$, and it can only be 
attained by setting $\mathcal{E}_2^{(j)}=\mathcal{E}_2^{(j-1)} \cup \{j\}$ (by the inductive hypothesis and the EDD rule). Therefore, $\mathcal{E}_2^{(j)}$ is a compact 
solution in this case as well.
\end{proof}

The following corollary is now straightforward from the facts that \begin{enumerate*}[label=(\roman*)]\item
the jobs 
are sorted according to the EDD order only once (prior to solving any of the subproblems), \item we can solve the $1|\U^{\Gamma}|\min_{\sigma \in \Sigma} \sum_{j=1}^n U_j(\sigma)$  problem by solving each of its $O(2^\ds{k_{u}})$ subproblems, and \item solving each subproblem by applying Algorithm~\ref{alg:hmext} requires $O(n)$ time.\end{enumerate*} 
\begin{corollary}
The problem $1|\U^{\Gamma}|\min_{\sigma \in \Sigma} \sum_{j=1}^n U_j(\sigma)$ can be solved in $O(n\log n+n2^\ds{k_{u}})$ time, where $\ds{k_{u}}$ 
stands for the number of jobs with a nonzero deviation.
\end{corollary}

\section{Summary and Discussion} 
\label{sumdis}
We investigated a \ds{single-machine} scheduling problem 
\ngtwo{with} uncertain 
job processing times. For each job, there are two potential processing time values: a nominal value and a deviated value\ngtwo{, 
and adopting the budgeted uncertainty approach,} 
the total number of deviations \ngtwo{is upper bounded} by a specified robustness parameter, $\Gamma$. Any 
realization of job processing times with no more than $\Gamma$ deviations is 
\ngtwo{a possible} 
scenario \ngtwo{of processing times.} 
We consider a variant of this problem where a job is considered \emph{early} if it is completed by its due date in all possible scenarios; otherwise, it is considered 
\ngtwo{tardy}. 
Our objective is to determine a solution that minimizes the number of tardy jobs. \ds{Since the problem is NP-hard, our primary goal is to analyze its parameterized complexity with respect to three natural parameters}: the robustness parameter, the number of 
\ngrev{distinct} due dates, and the number of jobs with nonzero deviations. We prove that the problem is W[1]-hard with respect to the robustness parameter, and we complement this result by designing an XP algorithm for it. Regarding the second parameter, we show 
that the problem remains hard even when the number of distinct due dates is only two, thus ruling out the existence of an XP algorithm (unless P$=$NP). We also prove that when there is only one distinct due date, the problem becomes solvable in polynomial time, {and when the number of 
\ngrev{distinct} due dates is 
\ngtwo{any} constant greater than one, the problem is solvable in pseudo-polynomial time.} For the 
\ngtwo{number of nonzero deviations}, we provide the strongest tractability result by proving that the problem is solvable in FPT time with respect to this parameter.

\ds{In the 
\ngrev{variant} we study, a job is classified as an early job only if its completion time is no later than its due date in all scenarios. This variant applies to settings in which a job must be rejected unless its timely completion can be guaranteed to the customer. In a different variant, studied by~\ngrev{\cite{Bougeret23}}, jobs must be processed regardless of their timing, and tardy jobs are penalized only after the scenario is realized. It is straightforward to show that our hardness results extend to this variant as well, whereas our positive results do not.} Consequently, it remains open whether 
\ngrev{the~\cite{Bougeret23}} 
variant is tractable (i.e., solvable in FPT time) when parameterized by the number of jobs with nonzero deviations. Similarly, its solvability in XP time with respect to the robustness parameter (or whether it is para-NP-hard\ngrev{, which would rule out an XP algorithm unless P=NP}) is also an open question. Beyond the parameters studied here, future research might consider others, including the number of unique nominal processing times (or processing time deviations) present in the problem instance, and the maximum due date value. Furthermore, examining the parameterized tractability of other scheduling problems with budgeted uncertainty presents another compelling direction. In essence, the confluence of parameterized complexity and robust scheduling remains a largely open and 
\ngrev{rich} area for ongoing research.   

\section*{Acknowledgements}
This research was supported in part by the Paul Ivanier Center for
Production Management, Ben-Gurion University of the Negev. 

\bibliographystyle{apalike-ejor}
\bibliography{bibliography}	 

\appendix

\end{document}